\documentclass[11pt,a4paper]{article}

\usepackage[margin=1in]{geometry}
\usepackage[T1]{fontenc}
\usepackage[utf8]{inputenc}
\usepackage{lmodern,microtype}
\usepackage{amsmath,amssymb,amsthm}
\usepackage{booktabs,tabularx,array}
\usepackage{tikz}
\usetikzlibrary{arrows.meta,positioning}
\usepackage[round,authoryear]{natbib}
\usepackage{xcolor}
\usepackage{hyperref}
\hypersetup{colorlinks=true,linkcolor=blue!45!black,
 citecolor=blue!45!black,urlcolor=blue!45!black}
\newcommand{\ind}{\mathbf{1}}
\newcommand{\res}{\mathrm{res}}
\newcommand{\all}{\mathrm{all}}
\newtheorem{proposition}{Proposition}

\title{When Token Transfers Do Not Mean Liquidity in Tokenized Real-World Assets}

\author{
Emre Atay T\"umer\textsuperscript{1}
\and
Rischan Mafrur\textsuperscript{2,3}
}

\date{
\small
\textsuperscript{1}Antalya Bilim University, Antalya, Türkiye\\
\textsuperscript{2}School of Business and Economics, Universiti Brunei Darussalam, Brunei Darussalam\\
\textsuperscript{3}School of CDMS, Western Sydney University, Australia\\[4pt]
\texttt{emreatay.tumer@std.antalya.edu.tr}\\
\texttt{rischan.mafrur@ubd.edu.bn}
}

\begin{document}
\maketitle

\begin{abstract}
On-chain transfer volume is often used to describe activity in tokenized real-world assets
(RWAs), although recorded transfers may reflect very different economic functions, including
token creation, destruction, issuer-related operations, and movements between other addresses.
Ownership measures face a similar problem because a large address balance may belong to a
pool or intermediary that represents multiple investors. We address these two measurement
problems by decomposing transfer activity and making ownership coverage explicit. Transfer
turnover is calculated before and after observable event categories are separated, using the
same asset, observation window, and supply denominator. For ownership, we compare
concentration measured against total supply with concentration among the addresses that remain
after identified intermediaries are excluded. Using Ethereum-based tokenized funds from
December 2025 to August 2026, we find that transfer classification can materially change
measured turnover. We also show that removing an intermediary can lower concentration when
measured against total supply while concentration among the remaining addresses increases.
The results demonstrate that transfer activity and address concentration need to be interpreted
with their underlying transaction functions and ownership coverage in mind. They provide
evidence about liquidity-related activity, but they do not directly measure market liquidity.
The code and data used in the analysis are available at
\url{https://github.com/atytmr07/rwa-liquidity}.
\end{abstract}

\noindent\textbf{Keywords:} real-world assets; tokenization; liquidity; transfer decomposition; ownership concentration; financial intermediation

\section{Introduction}

Tokenization promises to make traditional assets easier to trade. Treasury bills, money-market funds, and private credit can now be held as tokens on a public blockchain, where every transfer is visible to anyone. It is tempting to read this activity as evidence of liquidity: the more tokens move, the more liquid the asset appears. But a token moving is not the same as a token being traded. When a fund issues new shares, pays interest, or redeems an investor, the blockchain records a transfer just as it does when two investors trade. In our August 2026 observation for USYC, for example, 95.5\% of recorded transfer
volume falls outside the residual category once creation-like, destruction-like, and
issuer-linked events are separated. This does not imply that the remaining transfers are
secondary-market trades; rather, it shows how strongly a transfer-based activity measure can
depend on the economic function assigned to recorded movements.

A related problem arises in ownership measurement. A blockchain identifies addresses, not
beneficial owners. A single lending pool or custodial contract may hold tokens on behalf of
many investors, while one investor may control several addresses. Excluding a known
intermediary from an address-level concentration measure therefore changes both the observed
ownership distribution and the share of outstanding supply represented by that distribution.
A lower concentration statistic after such an exclusion does not necessarily imply more
dispersed ownership among the remaining addresses.

These measurement questions connect to established research on liquidity and intermediation. Bid--ask spreads and price impact describe the costs of executing transactions \citep{amihud1986,glosten1985,amihud2002}. Turnover captures trading intensity and has been used as a liquidity proxy \citep{datar1998}. Search frictions, intermediary capacity, and funding conditions also affect the ability to transact \citep{duffie2005,brunnermeier2009}. Applying these concepts to tokenized assets requires an initial distinction between token movements, market transactions, and contractual investor exit.


This paper focuses on two measurement issues in the use of blockchain data for tokenized
assets. One is the interpretation of transfer activity. We compare total transfer turnover with
the turnover that remains after specified event categories are separated, using the same
contract, observation window, and supply denominator. Any difference between the two measures
therefore comes from the event classification itself. The other issue concerns ownership
concentration when intermediary addresses are excluded. We report the share of total supply
represented by the retained addresses together with the concentration of balances within that
group. This allows us to separate changes in ownership coverage from changes in the
distribution of balances among the addresses that remain.

The framework extends the aggregate evidence in \citet{mafrur2026}, which examines transfer activity and participation across Ethereum-based real-world assets and recognises the inclusion of issuance and operational movements in transfer measures. This study develops an event-level treatment of that measurement problem and adds a formal analysis of intermediary holdings. 

\section{Related literature}

\subsection{Liquidity, trading activity, and intermediation}

Liquidity is multidimensional. \citet{amihud1986} relates transaction costs to asset valuation, while \citet{glosten1985} shows how asymmetric information can generate bid--ask spreads. Empirical proxies capture different aspects of this concept. \citet{datar1998} studies turnover as a liquidity proxy, and \citet{amihud2002} measures illiquidity through the relationship between absolute returns and trading volume. The use of token transfers introduces an additional classification task because the observed volume can contain transactions outside secondary-market exchange.

Market organisation also determines the conditions under which investors can trade. \citet{duffie2005} analyses search and bargaining in over-the-counter markets. \citet{brunnermeier2009} connects market liquidity to the funding conditions faced by intermediaries. These contributions motivate attention to the mechanisms supporting exchange, beyond the frequency of recorded movements. In a tokenized market, transferability, access to counterparties, and executable trading conditions remain separate empirical objects.

Redemption provides another route to investor exit. Evidence from corporate bond funds shows that portfolio illiquidity and investor flows can interact in ways that create fragility \citep{goldstein2017}. For tokenized funds, an assessment of exit liquidity therefore requires information about eligibility, fees, settlement timing, and restrictions on redemption. Low transfer turnover can coexist with an effective redemption mechanism, while frequent token movements can occur without deep secondary markets.

\subsection{Tokenization and financial market structure}

Blockchain research examines changes in recordkeeping, governance, and contracting. \citet{yermack2017} discusses the implications of distributed ledgers for corporate governance, and \citet{cong2019} analyses how blockchain consensus and smart contracts affect economic organisation. These mechanisms can alter the administration and transfer of financial claims. Their effects on liquidity depend on the surrounding market structure.

Security tokens also retain institutional and legal features that shape investor access. \citet{lambert2022} studies security token offerings and their financing characteristics. The BIS and CPMI emphasise the arrangements governing tokenized assets throughout their life cycle \citep{bis2024}. Within decentralised finance, composability permits assets to enter lending protocols, pools, and other contractual arrangements \citep{schar2021}. Such uses generate observable transfers and concentrated contract balances, making institutional identification central to empirical measurement.

Empirical studies of tokenized real estate provide direct evidence on ownership and trading. \citet{swinkels2023} examines ownership distributions and liquidity across real estate tokens. \citet{kreppmeier2023} investigates security token offerings and subsequent secondary-market activity. Together, these studies motivate analysis at the level of the instrument and its trading arrangements. Tokenized products differ in how, and whether, they trade.

\subsection{Positioning within RWA measurement research}

\citet{mafrur2026} analyses nine Ethereum-based assets over December 2025--May 2026 using monthly transfer and participation measures. That study establishes an empirical baseline across asset categories and discusses the ambiguity of transfer-based liquidity proxies. Our framework addresses the construction of those proxies by decomposing activity within each asset--window observation and making intermediary coverage explicit.

The closest transaction-level study is \citet{luo2026}, who examine tokenized U.S. Treasuries
through functional transaction analysis and address-role inference. Their classification
distinguishes issuance, redemption, transfers, and bridging. We use event classification for a
different purpose. Our analysis measures how much a transfer-based turnover indicator changes
when specified event categories are separated while the asset and observation window remain
fixed. We also examine how removing identified intermediary addresses affects ownership
coverage and concentration. In this way, the paper focuses on how transaction classification
and intermediary treatment change the interpretation of commonly used on-chain indicators.

\section{Conceptual framework and data structure}

\subsection{Units of observation}

The framework uses three linked units of observation: the transfer event, the address balance, and the asset--window. An asset is identified by its blockchain and contract address. A transfer event records the transaction hash, log index, block, timestamp, sender, recipient, and quantity. An address balance records holdings at a specified block. The asset--window links transfers over a defined interval to a consistent supply or valuation denominator.

These units support two complementary analyses. The activity analysis assigns an economic interpretation to observable event patterns and measures the volume remaining under specified classifications. The ownership analysis evaluates the distribution of balances and the proportion of supply represented after intermediary treatment. Figure~\ref{fig:framework} presents their integration.

\begin{figure}[htbp]
\centering
\begin{tikzpicture}[
 box/.style={draw,rounded corners,align=center,font=\small,
 text width=5.3cm,minimum height=1cm,inner sep=5pt},
 arrow/.style={-{Latex[length=2mm]},thick},
 node distance=0.65cm and 0.7cm]
\node[box] (events) {Transfer events\\Event conventions and configured roles};
\node[box,right=of events] (balances) {Address balances\\Supply and intermediary holdings};
\node[box,below=of events] (activity) {Activity decomposition\\Total and residual turnover};
\node[box,below=of balances] (ownership) {Ownership decomposition\\Coverage and conditional concentration};
\node[box,below=1cm of activity,xshift=3.15cm,text width=8cm] (interpretation)
 {Integrated interpretation\\Circulation, participation, and investor exit};
\draw[arrow] (events)--(activity);
\draw[arrow] (balances)--(ownership);
\draw[arrow] (activity)--(interpretation);
\draw[arrow] (ownership)--(interpretation);
\end{tikzpicture}
\caption{Structure of the proposed measurement framework. Event classification and ownership coverage provide complementary information for interpreting tokenized-asset activity.}
\label{fig:framework}
\end{figure}
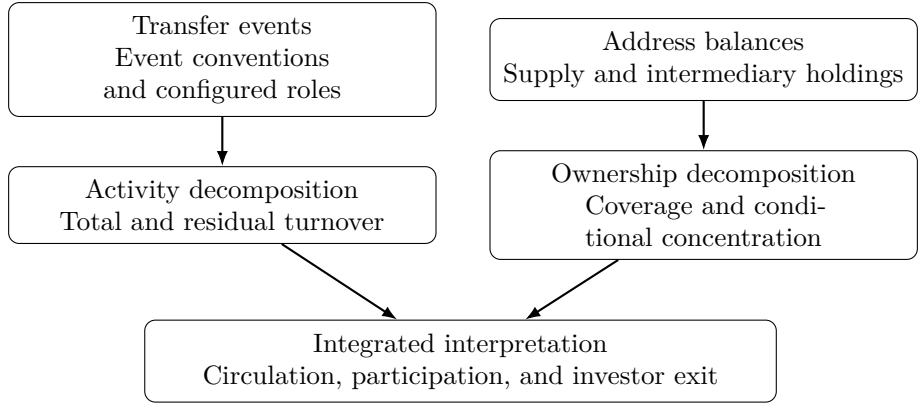

\subsection{Event interpretation}

The framework distinguishes observable address conventions from economic
attribution. ERC-20 recommends that token creation be recorded by a transfer
from the zero address, while no corresponding burn convention is prescribed
\citep{erc20}. Some contracts instead use designated sink addresses.
The implementation therefore treats the zero address and the configured sink
as boundary addresses for event classification. The resulting categories are
descriptive labels; claims about actual issuance, cancellation, or redemption
require additional evidence from contract behaviour or issuer documentation.

The implementation assigns events to five mutually exclusive categories. Let
$\mathcal{B}$ denote the set of designated boundary addresses, consisting of the
zero address and the configured sink address
\texttt{0x000000000000000000000000000000000000dEaD}. An event with both
endpoints in $\mathcal{B}$ is classified as unresolved. Otherwise, an event
originating in $\mathcal{B}$ is classified as creation-like, while an event
terminating in $\mathcal{B}$ is classified as destruction-like.

These labels describe the implemented address convention and do not by
themselves establish economic issuance or cancellation. In particular, the
creation-like category identifies transfers originating from a configured
boundary address; verification of actual token creation requires examination
of the relevant contract behaviour and supply changes.

These rules precede the issuer rule. An event that has not already been
classified as creation-like, destruction-like, or unresolved is classified as
issuer-linked when at least one endpoint belongs to the asset's configured
issuer-address list. An address enters that list only when the asset's own
issuer documents it in a published list of contracts or it carries the
issuer's label on Etherscan; transaction behaviour alone does not qualify an
address. All remaining events are classified as residual.

Table~\ref{tab:taxonomy} summarises the classification. An issuer-linked movement does not, by itself, establish a subscription or cash redemption. Similarly, residual transfers may record investor exchange, custody movements, collateral transfers, or wallet reorganisation. The unresolved category captures conflicts in the implemented rules; it does not encompass every event whose economic purpose remains unknown.

\begin{table}[htbp]
\centering
\small
\caption{Implemented event categories and limits of interpretation}
\label{tab:taxonomy}
\begin{tabularx}{\textwidth}{@{}p{2.6cm}XX@{}}
\toprule
Category & Observable rule & Interpretation \\
\midrule
Creation-like
& Sender in $\mathcal{B}$, recipient outside $\mathcal{B}$
& Boundary-address outflow; economic issuance and supply effect require verification \\

Destruction-like
& Recipient in $\mathcal{B}$, sender outside $\mathcal{B}$
& Boundary-address inflow; cancellation or redemption is not established \\

Issuer-linked
& Neither endpoint in $\mathcal{B}$; at least one endpoint in the configured issuer list
& Movement involving an address documented or labelled by the asset's own issuer, such as a redemption contract \\

Residual
& No preceding rule applies
& Remaining transfer activity, potentially including operational movements \\

Unresolved
& Both endpoints in $\mathcal{B}$
& Event crosses between configured boundary addresses and cannot be assigned uniquely to the preceding categories \\
\bottomrule
\end{tabularx}
\end{table}

\subsection{Data scope and validation}

Our panel contains 144 asset--month rows for 16 Ethereum contracts over December 2025--August 2026. PAXG observations are missing because the collection procedure exceeded its 250,000-log ceiling. XAUT, which is part of the sample in \citet{mafrur2026}, lies beyond the same ceiling and is not included; both gold tokens would require an indexed data source. Of the remaining 135 observations, 18 concern USDM and STBT, whose reconstructed balances failed the read-block supply check. Their ownership measures are therefore withheld in the exported panel. Their turnover denominators also depend on the failed reconstruction, so those observations are not used in the empirical analysis.

The focal empirical analysis concentrates on BUIDL, USYC, OUSG, and FDIT. These four tokenized funds satisfy the data-quality conditions required for the matched comparisons developed in Section~\ref{sec:methodology}: their transfer reconstruction is complete, their closing balance distributions pass the read-block supply reconciliation, and no dropped transfers are recorded. Restricting the focal analysis to observations satisfying these conditions reduces the risk that differences between total and residual turnover, or between alternative concentration measures, are driven by incomplete event coverage or failed balance reconstruction.

We use calendar months with UTC boundaries. The August window is
$(\text{1 August 2026},\text{1 September 2026})$, covering 31 days and Ethereum
blocks 25,656,293--25,878,704, inclusive. The collection is pinned to block
26,086,416, timestamped 29 September 2026 at 23:29:23 UTC. Each transfer event is identified by its contract address, transaction hash, and
log index, so that no event is counted twice. Transfer and balance data are obtained from Ethereum JSON-RPC. The panel is produced by a dedicated \texttt{paper}
command that applies the five-category classification in
Table~\ref{tab:taxonomy}. We validate each reconstructed ledger by comparing the sum of address balances
with the contract's \texttt{totalSupply()} at the read block.

\section{Measurement methodology}
\label{sec:methodology}
\subsection{Transfer decomposition}

Let $E_{it}$ denote the observed events for asset $i$ in the interval $[t_0,t_1)$, and let $q_e\geq0$ denote the quantity transferred in event $e$. Partition the events into creation-like, destruction-like, issuer-linked, residual, and unresolved categories, indexed by $k\in\{m,b,o,r,u\}$. Category volume is
\begin{equation}
 V^k_{it}=\sum_{e\in E_{it}}q_e\ind\{c_e=k\},
\end{equation}
where $c_e$ denotes the event category. Total transfer volume satisfies
\begin{equation}
 V^{\all}_{it}=V^m_{it}+V^b_{it}+V^o_{it}+V^r_{it}+V^u_{it}.
 \label{eq:partition}
\end{equation}
This is a decomposition of recorded movements. A sequence of transfers associated with one economic transaction may contribute several events to the total.

For positive end-of-window supply $S_{it}$, define
\begin{equation}
 T^{\all}_{it}=\frac{V^{\all}_{it}}{S_{it}},
 \qquad
 T^{\res}_{it}=\frac{V^r_{it}}{S_{it}}.
 \label{eq:turnover}
\end{equation}
Using the same denominator isolates the classification effect on the numerator. These measures describe turnover in native token units. Dollar-based turnover requires compatible event valuations and an asset-value denominator. The two constructions need not coincide when prices or net asset values change during the observation window.

For positive residual volume, the classification factor is
\begin{equation}
 F_{it}=\frac{T^{\all}_{it}}{T^{\res}_{it}}
       =\frac{V^{\all}_{it}}{V^r_{it}}.
 \label{eq:factor}
\end{equation}
For positive total volume, the share outside the residual category is
\begin{equation}
 X_{it}=1-\frac{V^r_{it}}{V^{\all}_{it}}.
 \label{eq:excluded}
\end{equation}
A high $F_{it}$ indicates that the observed turnover measure is sensitive to event composition. The category shares $V^k_{it}/V^{\all}_{it}$ identify the source of that sensitivity. Unresolved volume is reported separately so that incomplete attribution remains visible.

If residual volume is zero and total volume is positive, $X_{it}=1$ and the classification factor is undefined. If both are zero, the observation records no detected movement over a completely observed window. Missing event coverage is a separate data condition. If total volume is zero, $X_{it}$ is also undefined. Positive self-transfers remain in transfer volume under the implemented convention, although they are excluded from participation below; a sensitivity analysis can remove them from both measures.

\subsection{Participation and balance dormancy}

Participation measures depend on the population used in the denominator. Let $\mathcal{H}_{it}$ be the set of addresses with positive balances at the window end. Define the set of addresses participating in positive-value residual transfers as
\begin{equation}
 \mathcal{A}_{it}=\left\{a:\exists e\in E_{it},\ c_e=r,\ q_e>0,
 f_e\neq d_e,\ a\in\{f_e,d_e\}\right\},
\end{equation}
where $f_e$ and $d_e$ denote sender and recipient. Null and designated burn addresses are excluded. The positive-value and non-self-transfer conditions prevent zero-value messages and self-transfers from establishing participation. 

The active share of end-of-window holders is
\begin{equation}
 P_{it}=\frac{|\mathcal{A}_{it}\cap\mathcal{H}_{it}|}
 {|\mathcal{H}_{it}|}.
\end{equation}
This bounded measure differs from $|\mathcal{A}_{it}|/|\mathcal{H}_{it}|$, which can exceed one when transacting addresses exit before the end of the window. Both quantities can be informative, provided their populations are stated. We report $P_{it}$ only for a valid closing holder distribution with $|\mathcal{H}_{it}|>0$; otherwise it is undefined.

For validated balances $b_{ait}$ summing to supply, balance dormancy is
\begin{equation}
 D_{it}=\frac{\sum_{a\in\mathcal{H}_{it}}b_{ait}
 \ind\{a\notin\mathcal{A}_{it}\}}{S_{it}}.
\end{equation}
This statistic measures the share of closing balances held at addresses without observed residual participation. It is sensitive to the observation window and to the visibility of activity within pooled arrangements. We calculate participation and dormancy over the full closing holder population; the intermediary exclusions introduced below apply to retained concentration, not to these activity measures.

\subsection{Concentration and retained ownership coverage}

Let $s_{ait}=b_{ait}/S_{it}$ be the share of outstanding supply held by address $a$. For a complete and consistent balance distribution, address concentration is
\begin{equation}
 H^{\all}_{it}=10{,}000\sum_a s_{ait}^{2}.
\end{equation}
Let $\mathcal{P}_{it}$ denote the set of intermediary addresses excluded from a particular ownership analysis. The share of supply retained is
\begin{equation}
 C_{it}=\sum_{a\notin\mathcal{P}_{it}}s_{ait}.
 \label{eq:coverage}
\end{equation}
The contribution of retained addresses to full-supply concentration is
\begin{equation}
 H^{\mathrm{ret}}_{it}=10{,}000
 \sum_{a\notin\mathcal{P}_{it}}s_{ait}^{2}.
 \label{eq:retained}
\end{equation}
For $C_{it}>0$, concentration conditional on retained holdings is
\begin{equation}
 H^{\mathrm{cond}}_{it}=10{,}000
 \sum_{a\notin\mathcal{P}_{it}}
 \left(\frac{s_{ait}}{C_{it}}\right)^2.
 \label{eq:conditional}
\end{equation}

\begin{proposition}[Coverage decomposition]
For any non-empty retained balance distribution with $C_{it}>0$,
\begin{equation}
 H^{\mathrm{ret}}_{it}=C_{it}^{2}H^{\mathrm{cond}}_{it}.
 \label{eq:identity}
\end{equation}
For two observations with positive coverage and concentration,
\begin{equation}
 \Delta\log H^{\mathrm{ret}}_{it}
 =2\Delta\log C_{it}+\Delta\log H^{\mathrm{cond}}_{it}.
 \label{eq:change}
\end{equation}
\end{proposition}

\begin{proof}
Factoring $C_{it}^{-2}$ out of the sum in Equation~\eqref{eq:conditional} gives $H^{\mathrm{cond}}_{it}=H^{\mathrm{ret}}_{it}/C_{it}^{2}$. Taking logarithms and differences yields Equation~\eqref{eq:change}.
\end{proof}

The decomposition separates two sources of movement in the retained-address statistic. The first is a change in the fraction of supply represented by the retained addresses. The second is a change in the distribution of balances within that population. A fall in retained-address concentration can therefore coincide with increasing conditional concentration when retained coverage declines sufficiently.

The interpretation of top-$k$ shares follows the same principle. Let $K^{\mathrm{ret}}_{k,it}$ be the sum of the largest $k$ retained address balances divided by total supply, with addresses ranked again after exclusion. The corresponding conditional share is
\begin{equation}
 K^{\mathrm{cond}}_{k,it}=\frac{K^{\mathrm{ret}}_{k,it}}{C_{it}}.
\end{equation}
Reporting coverage alongside both concentration conventions makes the denominator choice explicit. None of these address-level measures resolves beneficial ownership within an excluded pool or common ownership across retained addresses.

When retained coverage falls below one half, the retained addresses represent a minority of supply, and conditional measures describe that minority among itself. The panel therefore withholds $H^{\mathrm{cond}}_{it}$ and $K^{\mathrm{cond}}_{k,it}$ when $C_{it}<0.5$ and reports them as not assessable. In our panel this applies to every ZTLN observation, where a single retained address holds one third of supply after a pooling contract is excluded, and to USYC in December 2025, when coverage is 38.4\%.

\subsection{Matched comparisons and classification sensitivity}

The primary comparison is within an asset--window. Contract identity, event coverage, observation boundaries, and the supply denominator remain fixed while the event classification changes. A baseline specification removes events meeting the creation-like and destruction-like rules. An extended specification additionally separates documented issuer and operational movements. A further specification can incorporate independently identified custody, bridge, or collateral activity.

The difference between two specifications quantifies sensitivity to the additional classification rule. Address labels require evidence and an applicable time interval, since an address can change function. The register contains explanatory notes but does not implement time-varying role assignments; applying those labels retrospectively therefore imposes an additional assumption. At the ownership level, sensitivity is assessed by reporting $C_{it}$, $H^{\mathrm{ret}}_{it}$, and $H^{\mathrm{cond}}_{it}$ for each intermediary treatment. These comparisons identify measurement effects without assigning a causal interpretation to differences between assets.

\section{Empirical Analysis}
\subsection{Sensitivity of transfer turnover}

We begin with four tokenized funds for which transfer reconstruction is complete and the
closing balance distribution passes the read-block supply reconciliation described in
Section~3.3. Table~\ref{tab:turnover} reports their August 2026 observations.
Within each asset--month, total and residual turnover use the same observation interval and
closing-supply denominator. Differences between the two measures therefore arise from the
event classification applied to the numerator. Issuer-linked events follow the rule in Section~3.2. For OUSG the list contains five contracts that Ondo Finance publishes for subscriptions and redemptions \citep{ondoaddresses}; for USYC it contains the two Teller contracts currently published by Circle
together with two earlier Teller contracts documented in archived versions of
the same smart-contract page \citep{circleusyc}. No documented or labelled issuer address was found for BUIDL or FDIT. USYC's Tellers take part in no August transfer, so issuer-linked volume is non-zero only for OUSG. Unresolved volume is zero for all four funds.

\begin{table}[htbp]
\centering
\small
\caption{Transfer-turnover decomposition, August 2026}
\label{tab:turnover}
\begin{tabular}{@{}lrrrr@{}}
\toprule
Asset & $T^{\all}$ & $T^{\res}$ & $F$ & $X$ (\%) \\
\midrule
BUIDL & 0.1940 & 0.1044 & 1.9 & 46.2 \\
USYC & 10.9039 & 0.4927 & 22.1 & 95.5 \\
OUSG & 0.5286 & 0.0563 & 9.4 & 89.4 \\
FDIT & 1.0478 & 0.6939 & 1.5 & 33.8 \\
\bottomrule
\end{tabular}

\vspace{0.5em}
\begin{minipage}{0.97\textwidth}
\footnotesize
All observations cover $(1\text{ August}-1\text{ September})$ 2026 in UTC. Turnover is native-token transfer volume divided by reconstructed closing supply, expressed as a ratio. Residual volume excludes creation-like, destruction-like, and issuer-linked events. $F=V^{\all}/V^r$ and $X=100(1-V^r/V^{\all})$. Values are calculated from unrounded panel entries and rounded only for display. Source: authors' calculations from Ethereum data (panel.csv in the replication package); the limitations of historical reconstruction are described in Section~3.3.
\end{minipage}
\end{table}

BUIDL's total turnover of 0.1940 falls to 0.1044 after the baseline exclusions, which account for 46.2\% of recorded volume. USYC has the largest classification factor among these four observations: total turnover of 10.9039 and residual turnover of 0.4927 imply a factor of 22.1, with 95.5\% of volume outside the residual category. The corresponding factors for OUSG and FDIT are 9.4 and 1.5. These differences describe the composition of token movements, not differences in execution costs or market depth.

OUSG provides the clearest example of the effect of issuer-linked classification. The verified source of Ondo's instant manager implements redemption by transferring OUSG from the investor to the manager and subsequently burning it \citep{ousgmanager}. Under the issuer rule the first leg is issuer-linked and the second destruction-like; under a baseline that removes only creation-like and destruction-like events, the first leg would remain residual. Ondo also lists an OUSG recipient address for manual redemptions \citep{ondoaddresses}. Under that baseline, OUSG's August residual turnover would be 0.2756, with $F=1.9$ and $X=47.9\%$. Classifying the documented addresses as issuer-linked reduces residual turnover to 0.0563 and raises $F$ to 9.4 and $X$ to 89.4\%. In August, the two managers account for 55.3\% and the recipient for 24.2\% of baseline residual volume; over December 2025--August 2026 the shares are 31.0\% and 55.5\%.

Table~\ref{tab:issuer} applies the same comparison to every asset in
the panel with at least one observed issuer-linked event during the sample
period. Besides OUSG and USYC, these are USTB, whose RedemptionIdle contract Superstate publishes \citep{superstate}; mTBILL, whose issuance and redemption vaults Midas publishes together with three further redemption vaults listed in the same Midas registry; TBILL, whose own token contract both receives TBILL from holders and burns it; and CANA, whose distribution address carries the issuer's Etherscan label. 

\begin{table}[htbp]
\centering
\caption{Issuer-linked classification effects}
\label{tab:issuer}
\small
\begin{tabular}{lrrrr}
\toprule
Asset &
Issuer-linked events &
Baseline residual removed &
August $X_{\mathrm{base}}$ &
August $X$ \\
&
Dec.~2025--Aug.~2026 &
Dec.~2025--Aug.~2026 (\%) &
(\%) &
(\%) \\
\midrule
OUSG   & 185   & 86.5 & 47.9 & 89.4 \\
TBILL  & 117   & 41.8 & 56.4 & 78.7 \\
USTB   & 4,724 & 26.5 & 70.2 & 80.8 \\
mTBILL & 730   & 7.1  & 78.3 & 78.3 \\
USYC   & 11    & 0.0  & 95.5 & 95.5 \\
CANA   & 9     & 0.0  & 0.0  & 0.0  \\
\bottomrule
\end{tabular}

\begin{minipage}{0.95\textwidth}
\footnotesize
Issuer-linked event counts and the share of baseline residual volume removed
by the issuer-linked rule are calculated over December 2025--August 2026. The baseline
classification removes creation-like and destruction-like events but does not separately
classify issuer-linked movements. The final two columns report $X$ for August 2026 under the
baseline and extended classifications, respectively. Shares below 0.05\% are displayed as
0.0\%. Source: authors' calculations from \texttt{paper/panel.csv}.
\end{minipage}
\end{table}

The classification effect also varies materially over time. BUIDL's August factor of 1.9 lies close to its nine-month median of 2.1, but the factor reached 12.6 in July 2026. Table~\ref{tab:months} reports the median and range of $X$ and $F$ across the nine months for the four funds. $X$ remains between 0 and 100\%, whereas $F$ becomes large whenever residual volume is small relative to total volume, as OUSG's maximum of 196.4 shows. $X$ is therefore the more stable summary of event composition.

\begin{table}[htbp]
\centering
\small
\caption{Event composition across December 2025--August 2026}
\label{tab:months}
\begin{tabular}{@{}lrrrrr@{}}
\toprule
Asset & Months & Median $X$ (\%) & Range $X$ (\%) & Median $F$ & Range $F$ \\
\midrule
BUIDL & 9 & 51.9 & 43.5--92.0 & 2.1 & 1.8--12.6 \\
USYC & 9 & 74.4 & 65.8--95.5 & 3.9 & 2.9--22.1 \\
OUSG & 9 & 97.5 & 77.3--99.5 & 39.7 & 4.4--196.4 \\
FDIT & 9 & 48.9 & 5.2--65.1 & 2.0 & 1.1--2.9 \\
\bottomrule
\end{tabular}

\vspace{0.5em}
\begin{minipage}{0.97\textwidth}
\footnotesize
Months counts calendar months with positive transfer volume. Medians and ranges are taken over the months in which each ratio is defined. The source is the \texttt{paper/panel.csv}.
\end{minipage}
\end{table}

\subsection{Intermediary exclusion and concentration}

At the end of August 2026, the identified Flux Finance fOUSG contract holds 27.9\% of OUSG's reconstructed supply. Excluding it reduces the top-ten share from 96.8\% to 69.6\% and the HHI from 1,680 to 901 when total supply remains the denominator. The retained addresses cover 72.1\% of supply. Conditional on their own holdings, their top-ten share is 96.6\% and their HHI is 1,734. The HHI reduction equals the excluded address's contribution; the change in the top-ten share also reflects the entry of the next-ranked retained address. Neither change identifies the ownership distribution inside the vault.

\begin{table}[htbp]
\centering
\small
\caption{OUSG concentration after excluding the fOUSG contract}
\label{tab:coverage}
\begin{tabular}{@{}lrrrrr@{}}
\toprule
Month end & $C$ (\%) & $K^{\mathrm{ret}}_{10}$ (\%) & $K^{\mathrm{cond}}_{10}$ (\%) & $H^{\mathrm{ret}}$ & $H^{\mathrm{cond}}$ \\
\midrule
December 2025 & 94.6 & 72.0 & 76.1 & 654 & 731 \\
August 2026 & 72.1 & 69.6 & 96.6 & 901 & 1,734 \\
\bottomrule
\end{tabular}

\vspace{0.5em}
\begin{minipage}{0.97\textwidth}
\footnotesize
Closing balances are reconstructed at each month end. The same fOUSG contract is excluded in both observations. $C$ is retained supply coverage; conditional measures use retained balances as their denominator. HHI is on the 0--10,000 scale. Identities hold for unrounded values, so ratios reconstructed from the rounded entries may differ slightly.
\end{minipage}
\end{table}

Table~\ref{tab:coverage} shows why the denominator matters over time. From December 2025 to August 2026, coverage declines from 94.6\% to 72.1\%. The retained top-ten share of total supply falls from 72.0\% to 69.6\%, whereas the corresponding share within retained holdings rises from 76.1\% to 96.6\%. Thus, the full-supply statistic alone would obscure increasing concentration among retained addresses. This is an accounting comparison and does not establish that the vault's growth caused the change in the remaining distribution.

A numerical example isolates the same effect. Suppose a pooled contract holds two-thirds of supply and one direct address holds the remaining third. Address concentration is
\begin{equation}
 H^{\all}=10{,}000\left[\left(\frac{2}{3}\right)^2+
 \left(\frac{1}{3}\right)^2\right]=\frac{50{,}000}{9}\approx5{,}555.6.
\end{equation}
After removing the pool, retained concentration is $10{,}000/9\approx1{,}111.1$ and coverage is $1/3$. Conditional concentration is $10{,}000$, because one address holds all retained balances. The pooled balance's underlying ownership remains unspecified. Equation~\eqref{eq:change} generalises this denominator effect across observation dates.

\section{Discussion}

\subsection{Implications for liquidity measurement}

Before measuring turnover, one has to know what each transfer is. Total turnover describes the intensity of observed token movements. Residual turnover describes circulation outside identified excluded categories. Execution liquidity additionally depends on prices, transaction costs, available depth, and the ability to locate a counterparty. Connecting these dimensions requires transaction-level evidence on exchange or independent market data.

This distinction also affects comparisons across product types. A fund designed around subscriptions and redemptions may record limited circulation among existing holders. A transferable commodity token may circulate across several venues and custodial arrangements. Applying a common transfer measure without considering these mechanisms can combine economically different activities. Classification makes their composition visible and provides a basis for more specific comparisons.

Ownership concentration supplies complementary information. Concentrated address holdings can indicate a small direct holder base, substantial intermediation, or both. Reporting the retained supply share and conditional concentration separates these possibilities to the extent permitted by the observed ledger. Recovering beneficial ownership requires additional evidence about pooled positions and common control.

\subsection{Relationship to aggregate RWA evidence}

The event-level results clarify an ambiguity in the aggregate transfer measures reported by \citet{mafrur2026}. That study documents cross-asset differences in monthly activity and
participation but does not identify the economic function of individual transfer events.
The decomposition used here shows how those aggregate measures change when observable
creation-like, destruction-like, and documented issuer-linked movements are separated.

The framework also complements functional transaction analysis \citep{luo2026}. Detailed classification supplies the event categories, while the matched ratios quantify their implications for liquidity proxies. The ownership decomposition addresses a separate issue arising when tokenized claims enter pooled or custodial arrangements.

\subsection{Scope and limitations}

The empirical application covers a selected set of Ethereum contracts observed over aligned
calendar-month windows. It does not describe the full RWA market, so the results are limited to
the measurement behaviour observed in these instruments and should not be generalized to all
tokenized assets.

Data coverage is also affected by transfer intensity. The collection procedure reaches its
250,000-log ceiling for PAXG observations and for XAUT, which makes highly active contracts
more difficult to reconstruct with the current approach. This limits our ability to compare the
sample with heavily transferred commodity tokens and introduces a sample-selection issue
related to activity itself.

Event classification is incomplete when operational addresses are undocumented, address roles
change over time, or economic settlement occurs outside Ethereum. We apply a conservative
documentation rule. An address that appears to perform a distribution, custody, or settlement
function remains in the residual category unless that role is supported by issuer documentation
or the specified external label source. As a result, residual volume can still include
non-trading activity, while some economic trades may occur without a contemporaneous on-chain
transfer. Residual turnover is therefore a classification-dependent measure of observed
activity and should not be treated as a direct estimate of secondary-market trading volume or
as a formal upper or lower bound.

Historical balance reconstruction introduces another limitation. Rebasing, non-standard
issuance mechanisms, sink transfers that do not reduce reported supply, omitted events, and
intermediate reconstruction errors can cause reconstructed balances to diverge from outstanding
supply. Observations that fail the implemented supply reconciliation are excluded from the
empirical analysis.

Address concentration also remains an address-level measure, not a direct measure of beneficial
ownership. Removing an identified pool does not reveal how the balance inside that pool is
distributed across underlying investors, and several retained addresses may still be controlled
by the same entity. Reporting ownership coverage makes this uncertainty visible, but resolving it
would require additional off-chain ownership information.

\section{Conclusion}

Blockchain records make token transfers and address balances observable, but neither maps
directly into the economic concepts that researchers may wish to measure. Transfer volume can
combine issuance, destruction, issuer-related operations, custody movements, and transfers
between other addresses, while address balances can reflect both direct holdings and
intermediated positions. This paper develops a framework for making these measurement choices
explicit.

On the activity side, the framework compares total and residual turnover within the same
asset--window and with the same supply denominator. The empirical application shows that the
difference can be economically large. Across the four August 2026 fund observations,
total-to-residual turnover factors range from 1.5 to 22.1, and the effect varies materially
across months. For OUSG, treating documented redemption-related addresses as issuer-linked
raises the August classification factor from 1.9 to 9.4.

On the ownership side, the framework separates the share of supply represented by retained
addresses from concentration within that retained population. The OUSG example shows why the
distinction matters: declining retained coverage can reduce concentration measured against
total supply even while concentration among the remaining addresses increases.

Residual transfer volume is not a direct measure of secondary-market trading, just as address-level concentration is not a direct measure of beneficial ownership. Their interpretation depends on how transfer events are classified and how much of the token supply is represented in the ownership measure. Empirical studies should make these choices explicit when using on-chain activity or concentration as evidence on market liquidity and investor participation.
\section*{Acknowledgements}

The authors used ChatGPT (GPT-5.6 Sol, OpenAI) to assist with language polishing, grammatical correction, and improving the clarity of selected passages in the manuscript. The tool was not used to generate research data, perform the empirical analysis, or determine the study's findings and conclusions. The authors reviewed and take full responsibility for the final content of the manuscript.

\section*{Data and research materials}

The analysis uses the release of \texttt{rwa-liquidity} (package version 0.1.1), available at \url{https://github.com/atytmr07/rwa-liquidity}. The results correspond to release \texttt{v0.1.1} of the repository. The release includes \texttt{paper/panel.csv}, the Table~\ref{tab:turnover} fragment, a nine-month summary of $X$ and $F$ for every asset, and macros for the figures quoted in the text. The collection command recorded by the repository is
\begin{quote}\small\ttfamily
\texttt{rwa-liquidity paper -{}-first-month 2025-12 -{}-last-month 2026-08 -{}-head 26086416}
\end{quote}
Appendix~\ref{app:contracts} records the focal contracts and classification addresses.

\appendix
\section{Contract identities and empirical conventions}
\label{app:contracts}

The focal Ethereum token contracts are:
\begin{quote}\small
BUIDL: \nolinkurl{0x7712c34205737192402172409a8f7ccef8aa2aec}\\
USYC: \nolinkurl{0x136471a34f6ef19fe571effc1ca711fdb8e49f2b}\\
OUSG: \nolinkurl{0x1b19c19393e2d034d8ff31ff34c81252fcbbee92}\\
FDIT: \nolinkurl{0x48ab4e39ac59f4e88974804b04a991b3a402717f}
\end{quote}
The further contracts in Table~\ref{tab:issuer} are:
\begin{quote}\small
TBILL: \nolinkurl{0xdd50c053c096cb04a3e3362e2b622529ec5f2e8a}\\
USTB: \nolinkurl{0x43415eb6ff9db7e26a15b704e7a3edce97d31c4e}\\
mTBILL: \nolinkurl{0xdd629e5241cbc5919847783e6c96b2de4754e438}\\
CANA: \nolinkurl{0x01995a697752266d8e748738aaa3f06464b8350b}
\end{quote}
The OUSG retained-ownership calculation excludes the fOUSG contract \\
\nolinkurl{0x1dd7950c266fb1be96180a8fdb0591f70200e018}.
The issuer-linked addresses, with the source that documents each, are:
\begin{quote}\small
OUSG \citep{ondoaddresses}: instant manager \nolinkurl{0x93358db73b6cd4b98d89c8f5f230e81a95c2643a}; legacy instant manager \nolinkurl{0x2826989983e3a66f0622132d019c2ae173eb6a43}; OUSG recipient for manual redemptions \nolinkurl{0x72be8c14b7564f7a61ba2f6b7e50d18dc1d4b63d}; legacy managers \nolinkurl{0xf16c188c2d411627d39655a60409ec6707d3d5e8} and \nolinkurl{0x3501883a646f1f8417bcb62162372550954d618f}.\\
USYC \citep{circleusyc}: Teller \nolinkurl{0xee35f963bfc71b51ec95147f26c030d674ea30e6}; \\ Cross Chain Teller \nolinkurl{0x5dbeccecebcdc2ce3258f6e638373d2923560c7d}; 
earlier Teller \nolinkurl{0x5c73e1cfdd85b7f1d608f7f7736fc8c653513b7a} and \\ Cross Chain Teller \nolinkurl{0x5575a88bec5b47ce8d270f6a4f2418865f16afd5}, listed on the same page in September 2025 and January 2026 (Internet Archive).\\
USTB \citep{superstate}: RedemptionIdle \nolinkurl{0x4c21b7577c8fe8b0b0669165ee7c8f67fa1454cf}.\\
mTBILL \citep{midas}: issuance vault \nolinkurl{0x99361435420711723af805f08187c9e6bf796683}; instant redemption vault \nolinkurl{0x569d7dccbf6923350521ecbc28a555a500c4f0ec}; mEDGE, mMEV, and mBASIS redemption vaults \nolinkurl{0x9b2c5e30e3b1f6369fc746a1c1e47277396af15d}, \nolinkurl{0xac14a14f578c143625fc8f54218911e8f634184d}, and \nolinkurl{0x0d89c1c4799353f3805a3e6c4e1cbbb83217d123}.\\
TBILL: the token contract itself.\\
CANA: \nolinkurl{0xccadea5cc24204995a98daa056c37bd5207fd0c5}, labelled ``Maseer: Deployer'' on Etherscan.
\end{quote}
No documented or labelled issuer address was found for BUIDL or FDIT. The designated event endpoints are the zero address and
\nolinkurl{0x000000000000000000000000000000000000dead}.

Each transfer is counted once by log identity: a transaction containing multiple transfer events may contribute several observations. All transfer quantities are denominated in the token's native units. Closing supply is reconstructed immediately before the exclusive window endpoint.

\end{document}